\def \VersionAuthor {}
\ifdefined\VersionAuthor
	\newcommand{\AuthorVersion}[1]{#1}
	\newcommand{\SpringerVersion}[1]{}
\else
	\newcommand{\AuthorVersion}[1]{}
	\newcommand{\SpringerVersion}[1]{#1}
\fi

\documentclass[a4paper,10pt]{llncs}
\usepackage{amsmath,amssymb,amsfonts}
\usepackage{graphicx}
\usepackage[utf8]{inputenc}
\usepackage[english]{babel}

\usepackage[ruled,vlined,linesnumbered]{algorithm2e}

\usepackage{subcaption}

\usepackage{paralist} % inline lists

\usepackage{amsmath} % for \text et al., for ``cases'' environmentâ€¦
\usepackage{amssymb} % for \mathbb et al.

\ifdefined\VersionWithComments
	\usepackage{draftwatermark}
	\SetWatermarkText{draft}
	\SetWatermarkScale{15}
	\SetWatermarkColor[gray]{0.9}
\fi
\usepackage[svgnames,table]{xcolor}
\definecolor{darkblue}{rgb}{0.0,0.0,0.6}
\definecolor{darkgreen}{rgb}{0, 0.5, 0}
\definecolor{darkpurple}{rgb}{0.7, 0, 0.7}
\definecolor{darkblue}{rgb}{0, 0, 0.7}

\usepackage[capitalise,english,nameinlink]{cleveref} % load after algorithm2e and hyperref
\crefname{line}{\text{line}}{\text{lines}} % to remove the capital

\DeclareMathOperator*{\argmax}{arg\,max}
\DeclareMathOperator*{\argmin}{arg\,min}

\newcommand{\BTrue}{\textsf{true}}
\newcommand{\BFalse}{\textsf{false}}

\ifdefined \VersionLong
	\newcommand{\LongVersion}[1]{\ifdefined\VersionWithComments{\color{red!40!black}#1}\else#1\fi}
	
\else
	\newcommand{\LongVersion}[1]{\ifdefined\VersionWithComments{\color{black!40}#1}\fi}
	
\fi

\usepackage{tikz}
\usetikzlibrary{arrows,automata}
\tikzstyle{every node}=[initial text=]
\tikzstyle{location}=[rectangle, rounded corners, minimum size=12pt, draw=black, fill=blue!10, inner sep=2pt]
\tikzstyle{invariant}=[draw=black, dotted, inner sep=1pt] % xshift=1em, 

\definecolor{coloract}{rgb}{0.50, 0.70, 0.30}
\definecolor{colorclock}{rgb}{0.4, 0.4, 1}
\definecolor{colordisc}{rgb}{1, 0, 1}
\definecolor{colorloc}{rgb}{0.2, 0.2, 0.35}
\definecolor{colorparam}{rgb}{1, 0.6, 0.0}

\newcommand{\styleclock}[1]{\ensuremath{\textcolor{colorclock}{\mathit{#1}}}}
\newcommand{\styledisc}[1]{\ensuremath{\textcolor{colordisc}{\mathit{#1}}}}

\newcommand{\styleparam}[1]{\ensuremath{\textcolor{colorparam}{\mathit{#1}}}}

\newcommand{\cellHeader}[1]{\cellcolor{blue!20}\textbf{#1}}
\newcommand{\rowlHeader}{\rowcolor{blue!20}\bfseries}

\ifdefined\VersionAuthor
\else
\fi

\ifdefined \VersionWithComments
	\usepackage{marginnote}
	\newcommand{\marginX}{\marginnote{\huge{\quad\quad\textbf{!}\quad\quad}}}
	\newcommand{\ea}[1]{\mbox{}{\color{blue}\marginX{}\textbf{[\'Etienne}: #1]}}
	\newcommand{\lf}[1]{\mbox{}{\color{purple}\marginX{}\textbf{[Laurent}: #1]}}
	\newcommand{\rs}[1]{\mbox{}{\color{orange}\marginX{}\textbf{[Romain}: #1]}}
	\newcommand{\instructions}[1]{{\color{red}\marginX{}\textbf{[Instructions: ``#1'']}}}
	\newcommand{\reviewer}[2]{\mbox{}{\color{red}\marginX{}\textbf{[Reviewer #1}: ``#2'']}}
	\newcommand{\todo}[1]{\mbox{}{\color{red}{\marginX{}\textbf{TODO}\ifx#1\\\else:\ \fi #1}}} % here, ``\\'' stands for ``empty''
\else
	\newcommand{\instructions}[1]{}
	\newcommand{\ea}[1]{}
	\newcommand{\lf}[1]{}
	\newcommand{\rs}[1]{}
	\newcommand{\reviewer}[2]{}
	\newcommand{\todo}[1]{}
\fi

\newcommand{\Mabs}{\widehat{M}}
\newcommand{\grandn}{{\mathbb N}}

\newcommand{\imitator}{\textsf{IMITATOR}}
\newcommand{\uppaal}{\textsc{Uppaal}}

\newcommand{\styleVar}[1]{\ensuremath{\mathit{#1}}}
\newcommand{\higherIDreceived}{\styleVar{higherIDreceived}}
\newcommand{\SenderID}{\styleVar{SenderID}}

\newcommand{\styleConst}[1]{\ensuremath{\mathsf{#1}}}
\newcommand{\modeOn}{\styleConst{On}}
\newcommand{\modeOff}{\styleConst{Off}}
\newcommand{\periodMin}{\styleConst{period_\mathsf{min}}}
\newcommand{\periodMax}{\styleConst{period_\mathsf{max}}}
\newcommand{\jitterMin}{\styleConst{jitter_\mathsf{min}}}
\newcommand{\jitterMax}{\styleConst{jitter_\mathsf{max}}}
\newcommand{\statusCandidate}{\styleConst{Candidate}}
\newcommand{\statusFollower}{\styleConst{Follower}}
\newcommand{\statusLeader}{\styleConst{Leader}}

\ifdefined \VersionWithComments
 	\definecolor{colorok}{RGB}{80,80,150}
\else
	\definecolor{colorok}{RGB}{0,0,0}
\fi

\newcommand{\eg}{\textcolor{colorok}{e.\,g.,}} % \xspace
\newcommand{\ie}{\textcolor{colorok}{i.\,e.,}} % \xspace
\newcommand{\aka}{\textcolor{colorok}{a.k.a.}\xspace}

\begin{document}

\sloppy

\title{Verification of an industrial asynchronous leader election algorithm using abstractions and parametric model checking%
	\thanks{%
		\AuthorVersion{This is the author version of the manuscript of the same name published in the proceedings of the 20th International Conference on Verification, Model Checking, and Abstract Interpretation (VMCAI 2019).}
% 	The final version is available at \href{www.dx.doi.org/10.1109/ICECCS2018.2018.00010}{10.1109/ICECCS2018.2018.00010}.
		This work is partially supported by
		Institut Farman (ENS Paris-Saclay $\&$ CNRS),
		by the ANR national research program PACS (ANR-14-CE28-0002)
		and
		by ERATO HASUO Metamathematics for Systems Design Project (No.\ JPMJER1603), JST.
	}
}

\author{\'Etienne Andr\'e\inst{1,2,3} \and Laurent Fribourg\inst{4} \and Jean-Marc Mota\inst{5} \and Romain Soulat\inst{5}} % \inst{1,2}
% \date{\today{}}
\institute{Université Paris 13, LIPN, CNRS, UMR 7030, F-93430, Villetaneuse, France
\and JFLI, CNRS, Tokyo, Japan
\and National Institute of Informatics, Japan
\and LSV, ENS Paris-Saclay $\&$ CNRS $\&$ INRIA, U. Paris-Saclay
\and Thales Research and Technology, Palaiseau, France
}

% For all page numbers, except p.1
\pagestyle{plain}

\maketitle

% For page numbers on p.1
\thispagestyle{plain}

\ifdefined \VersionWithComments
	\textcolor{red}{\textbf{This is the version with comments. To disable comments, comment out line~3 in the \LaTeX{} source.}}
\fi

\begin{abstract}
	The election of a leader in a network is a challenging task, especially when the processes are asynchronous, \ie{} execute an algorithm with time-varying periods.
	Thales developed an industrial election algorithm with an arbitrary number of processes, that can possibly fail.
	In this work, we prove the correctness of a variant of this industrial algorithm.
	We use a method combining abstraction, the SafeProver solver, and a parametric timed model-checker.
	This allows us to prove the correctness of the algorithm for a large number~$p$ of processes ($p=5000$).
\end{abstract}


\instructions{
Submissions are restricted to 20 pages in Springer’s LNCS format, not counting references. Additional material may be placed in an appendix, to be read at the discretion of the reviewers and to be omitted in the final version.
}
%\ea{hello}
%\lf{hello}
%\rs{hello}

%%%%%%%%%%%%%%%%%%%%%%%%%%%%%%%%%%%%%%%%%%%%%%%%%%%%%%%%%%%%
%%%%%%%%%%%%%%%%%%%%%%%%%%%%%%%%%%%%%%%%%%%%%%%%%%%%%%%%%%%%
\section{Introduction}
%%%%%%%%%%%%%%%%%%%%%%%%%%%%%%%%%%%%%%%%%%%%%%%%%%%%%%%%%%%%
%%%%%%%%%%%%%%%%%%%%%%%%%%%%%%%%%%%%%%%%%%%%%%%%%%%%%%%%%%%%

Distributed systems, where entities communicate with each other, are booming in our societies.
Drones communicating with each other, swarms of various objects, intelligent cars…\ all may face communication and leadership issues.
Therefore, the algorithm that all entities execute should be verified.
Thales developed an industrial election algorithm with an arbitrary number of processes, that can possibly fail.
We cannot describe the code of the actual algorithm for confidentiality issues.
Therefore, we consider a modified variant of the algorithm.
% In this work, we prove the correctness of this variant.
%
This algorithm focuses on the election of a leader in a distributed system with a potentially large number of entities or \emph{nodes} in an \emph{asynchronous} environment.
Our main contribution is to perform a formal verification of the algorithm correctness for a large number of nodes.
By correctness, we mean the actual election of the leader after a fixed number of rounds.

We consider here a special form of the general leader election problem \cite{lynch96}:
we assume that, in the network,
all the processes (or \emph{nodes}) have a specific ID number, and they execute the same code 
({\em symmetry}) 
in order to agree which ID number is the highest one.
In the synchronous context where all processes communicate simultaneously, the problem is often solved using the ``Bully algorithm'' \cite{bully}.
In the asynchronous context where each process communicates with a specific period 
possibly subject to delay variation (jitter), the problem is much more difficult.
Periods can be all slightly different from each other, which makes the problem particularly complex.
For example, a classical distributed leader election protocol, where the nodes exchange data using broadcasting, was designed by Leslie Lamport \cite{Lamport98} in the asynchronous context.
The correctness of this algorithm  was proved mechanically many times using, \eg{} TLA$^+$ tool \cite{TLA+}, or, more recently, using the timed model checking tool \uppaal{}~\cite{Kim}. 
However, these automated proofs work only for a small number~$p$ of processes, typically for $p \leq 10$.
In this paper, we  present a technique to prove the correctness of such a distributed leader election using automated tools for a large number of nodes (\eg{} $p=5000$).
The principle of the method relies on the
%classical 
abstraction method consisting in viewing the network from the point of view of a specific (but arbitrary) node, say $node_i$, and considering the rest of the nodes of the network as an abstract environment interacting with $node_i$.
In this abstract model, two basic properties of the algorithm can be proven. % (P1 and P2).
However, in order to prove the full correctness of the leader election algorithm, we will need an auxiliary model, where some timing information is added to (a raw form of) the abstract model. Using this auxiliary timed model, we are able to prove an additional property %(P3)
of the leader election algorithm. 
Thanks to %(P1-P2-P3)
the three aforementioned properties added as assumptions,
we can then prove the full correctness of the leader election algorithm, using
	the bounded model checker SafeProver~\cite{Etienne:2017}
% 	an SMT solver \cite{SMT}
on the abstract model.
%(see Section \ref{ss:???} for details). 

The leader election algorithm we use is not Lamport's algorithm, but a simple asynchronous form of the Bully algorithm.
\ea{Removed: ``This variant of the Bully algorithm works in our asynchronous framework because we restrict ourselves to a specific framework of network structure and asynchronous form of communications.''}%
We consider a specific framework of network structure and asynchronous form of communications.
Basically, we assume that:
\begin{enumerate}
	\item the graph is complete (every node communicate with all the other ones).
	\item the communications are instantaneous (the time between the sending of a message and its reception is null),
	and the nodes % and the message
		exchange data via synchronous one-way unconditional value passing.
	\item the processes are {\em visibly faulty}, \ie{} they always execute the generic code of the algorithm, trying to elect the leader when they are 
	non-faulty (mode $\modeOn{}$), and do nothing when they are faulty (mode \modeOff{}). 
\end{enumerate}
%%We think that these assumptions (except perhaps the first one) can be relatively easily lifted, that our abstraction-based approach can be adapted in order to take into 
%account more general structures; we plan in future work to present 
%the mechanical proof of original Lamport's algorithm for a large number~$p$ of nodes along the same lines as those explained in this paper.

\subsection*{Relationship with Thales' actual algorithm}
As mentioned above, for confidentiality issue, we cannot reveal the original algorithm developed at Thales.
Nevertheless, it is in essence the same as the one we present.
Only the executed code has been modified.
In addition, the technique presented in this paper was designed for and applied to the original algorithm.
To summarize, we present exactly the methodology, up to the content of the $UpdateNode$ code (that is still similar in spirit).

After its verification using the techniques we present here, the original algorithm has been implemented in~C, and is nowadays running in one of the Thales products.
This product embeds a standard processor (in the line of Intel~X86), with some limited RAM, hard drive, Ethernet ports, etc.

\subsection*{Related work}
The method proposed here makes use of several powerful techniques
such as counter abstraction,
	bounded model checking
	and
	parametric timed model checking
% and SMT solving \cite{SMT}
for verifying distributed fault-tolerant algorithms,
similarly to what has been recently described, \eg{}  in \cite{KVW15}. 
% applied
%to symmetric systems with a large (even sometimes parametric) number of identical processes. 
As said in \cite{KVW15}: ``Symmetry allows us to change representation
into a {\em counter representation} (also referred to as `counter abstraction'): (...) Instead of recording which process is in which local state, we record for each local state, how many processes are in this state.  Thus, we need one counter per local state $\ell$, hence we have a fixed number of counters. A step by a process that goes from local state $\ell$ to local state $\ell'$ is modeled by decrementing the counter associated with $\ell$ and incrementing the counter associated with $\ell'$. When the number~$p$ of processes is fixed, each counter is bounded by~$p$.'' The work described in \cite{KVW15} makes use of SMT solvers \cite{SMT} in order to perform finite-state model checking of the abstracted model. 

Our work can be seen as a new application of such techniques 
to (a variant of) an industrial election algorithm.
Another originality is to combine counter abstraction, bounded model checking, with \emph{parametric timed} model checking.

% \todo{il faudra citer les networks of identical processes \cite{AJ03,ADRST11}}
In an orthogonal direction, the verification of identical processes in a network, \ie{} a unknown number of nodes running the same algorithm, has been studied in various settings, notably in the long line of work around regular model checking~\cite{FO97,BJNT00}, and in various settings in the \emph{timed} case~\cite{AJ03,ADRST11,ADRST16}.
However, the focus of that latter series of works is on decidability, and they do not consider real-world algorithms, nor do they have tools implementing these results.

Finally, the line of works around the Cubicle model-checker \cite{CGKMZ12,CDZ17,CDZ17b} performs parameterized verification of cache memory protocols, that is also parameterized in the number of processes.
However, timing parameters are not present in these works.

\ea{In final version: cite \cite{KLF18}, as they also use different techniques, and untimed/timed techniques (also Uppaal with TA + proofs)}

\subsection*{Outline}
The rest of the paper is organized as follows.
\cref{section:algo} introduces the variant of the leader election algorithm we consider.
\cref{section:direct} presents our direct verification method for a small number of nodes.
\cref{section:abstraction} presents our abstraction-based verification for a much larger number of nodes.
\cref{section:conclusion} concludes the manuscript and outlines perspectives.

%\section{Preliminaries}
%%%%%%%%%%%%%%%%%%%%%%%%%%%%%%%%%%%%%%%%%%%%%%%%%%%%%%%%%%%%
%%%%%%%%%%%%%%%%%%%%%%%%%%%%%%%%%%%%%%%%%%%%%%%%%%%%%%%%%%%%
%\subsection{Preliminaries}
%%%%%%%%%%%%%%%%%%%%%%%%%%%%%%%%%%%%%%%%%%%%%%%%%%%%%%%%%%%%
%%%%%%%%%%%%%%%%%%%%%%%%%%%%%%%%%%%%%%%%%%%%%%%%%%%%%%%%%%%%

%%%%%%%%%%%%%%%%%%%%%%%%%%%%%%%%%%%%%%%%%%%%%%%%%%%%%%%%%%%%
%\subsection{Problem: Leader election}
%%%%%%%%%%%%%%%%%%%%%%%%%%%%%%%%%%%%%%%%%%%%%%%%%%%%%%%%%%%%

%\todo{bien mettre en valeur le fait qu'on est en asynchrone (contrairement a  Bully)}

%\todo{certain nombre de processus, activations, graphe complet avec $N$ nouds, messages labelises, temporisation, probleme de l'incertitude sur les periodes}

%\cite{AHV93}

%%%%%%%%%%%%%%%%%%%%%%%%%%%%%%%%%%%%%%%%%%%%%%%%%%%%%%%%%%%%
%%%%%%%%%%%%%%%%%%%%%%%%%%%%%%%%%%%%%%%%%%%%%%%%%%%%%%%%%%%%
\section{An asynchronous leader election algorithm}\label{section:algo}
%%%%%%%%%%%%%%%%%%%%%%%%%%%%%%%%%%%%%%%%%%%%%%%%%%%%%%%%%%%%
%%%%%%%%%%%%%%%%%%%%%%%%%%%%%%%%%%%%%%%%%%%%%%%%%%%%%%%%%%%%

%%%%%%%%%%%%%%%%%%%%%%%%%%%%%%%%%%%%%%%%%%%%%%%%%%%%%%%%%%%%
%\subsection{Solution: Bully Romain}
%%%%%%%%%%%%%%%%%%%%%%%%%%%%%%%%%%%%%%%%%%%%%%%%%%%%%%%%%%%%

Thales recently proposed a leader election algorithm.
% As Thales' original algorithm cannot be presented for confidentiality issues, we redesign a similar algorithm.
This simple leader election algorithm is based on the classical Bully algorithm originally designed for the synchronous framework~\cite{bully}.
Basically, all nodes have an ID (all different), and the node with the largest ID must be elected as a leader.
% 	\ea{c'est là qu'il faut plus d'explications sur la finalité !}
This algorithm is asynchronous.
As usual, each node runs the same version of the code.
We cannot describe the code of the actual algorithm for confidentiality issues, and we therefore consider and prove a modified variant of Thales' original algorithm, described throughout this section.

%This implementation should be viewed as an illustration for the
%overall verification techniques that will be described in the following sections.

%%%%%%%%%%%%%%%%%%%%%%%%%%%%%%%%%%%%%%%%%%%%%%%%%%%%%%%%%%%%
\subsection{Periods, jitters, offset}
%%%%%%%%%%%%%%%%%%%%%%%%%%%%%%%%%%%%%%%%%%%%%%%%%%%%%%%%%%%%

\LongVersion{Let us first describe the context and the hypotheses.}
The system is a fixed set of $p$~nodes $\mathcal{N} = \{node_1,\ldots,node_p\}$, for some $p \in \grandn$.
Each node $node_i$ is defined by:
\begin{enumerate}
	\item its integer-valued \emph{ID} $node_i.id \in \grandn$,
	\item its rational-valued \emph{activation period} $node_i.per \in [\periodMin{}, \periodMax{}]$,
	\item its rational-valued \emph{first activation time} $node_i.start\in [0,node_i.per]$ (which can be seen as an \emph{offset}, with the usual assumption that the offset must be less than or equal to the period),
		and
	\item its rational-valued \emph{jitter} values represent a delay variation for each period belonging to $[\jitterMin{}, \jitterMax{}]$, which is a static interval defined for all nodes and known beforehand.
\end{enumerate}
Observe that all periods are potentially different (even though they are all in a fixed interval, and each of them remains constant over the entire execution), which makes the problem particularly challenging.
In contrast, the jitter is different at each period (this is the common definition of a jitter), and the jitter of node~$i$ at the $j$th activation is denoted by $jitter_i^j$.
The $j$th activation of node $node_i$ therefore takes place at time $t_i^j = t_i^{j-1} + node_i.per + jitter_i^j$.
We have besides:
$t_i^0 = node_i.start$.

The concrete values for the static timing constants are given in \cref{table:constants}.

%------------------------------------------------------------
\begin{table}[tb]
	\centering
	\caption{Constants (in ms)}
	\begin{tabular}{l | r}
		\hline
		\cellHeader{Constant} & \cellHeader{Value} \\
		\hline
		$\periodMin{}$ & $49$ \\
		\hline
		$\periodMax{}$ & $51$ \\
		\hline
		$\jitterMin{}$ & $-0.5$ \\
		\hline
		$\jitterMax{}$ & $0.5$ \\
		\hline
	\end{tabular}

	\label{table:constants}
\end{table}
%------------------------------------------------------------

%------------------------------------------------------------
\begin{example}\label{example:3}
	Assume the system is made of three nodes.
	Assume $node_1.per = 49$.
	Recall that a period is an arbitrary \emph{constant} in a predefined interval.
	Assume $node_1.start = 0$.
	
	Assume $node_2.per = 51$ and $node_2.start = 30$.
	
	Assume $node_3.per = 49$ and $node_3.start = 0.1$.

	Also assume the jitters for the first three activations of the nodes given in \cref{table:example}.
	
	%------------------------------------------------------------
	\begin{table}[tb]
		\centering
		\caption{Jitter values for \cref{example:3}}
			\begin{tabular}{l | r | r | r}
				\hline
				\rowlHeader{}
				& $jitter^1$ & $jitter^2$ & $jitter^3$ \\
				\hline
				Node 1 & $0.5$ & $-0.5$ & $0.5$ \\
				\hline
				Node 2 & $0$ & $0.1$ & $0$ \\
				\hline
				Node 3 & $0.1$ & $0.3$ & $0.5$ \\
				\hline
			\end{tabular}

		\label{table:example}
	\end{table}
	%------------------------------------------------------------
	
	We therefore have
		$t_1^0 = 0$,
		$t_1^1 = 49.5$,
		$t_1^2 = 97.5$,
		$t_1^3 = 147.5$,
		$t_2^0 = 30$,
		$t_2^1 = 81$,
		$t_2^2 = 132.1$,
		$t_2^3 = 183$,
		$t_3^0 = 0.1$,
		$t_3^1 = 48.6$,
		$t_3^2 = 98.4$,
		$t_3^3 = 147.6$.
	The first activations of the nodes are depicted in \cref{figure:chronogramme}.
	Due to both uncertain periods and the jitters, it can happen that, between two consecutive activations of a node, another node may not be activated at all:
		for example, between $t_3^0$ and $t_3^1$, node~$1$ is never activated, and therefore node~3 will not receive a message from node~1 during this interval.
	Conversely, between two consecutive activations of a node, another node is activated twice:
		for example, between $t_3^1$ and $t_3^2$, node~$1$ is activated twice (\ie{} $t_1^1$ and $t_1^2$), and therefore node~3 may receive two messages from node~1 during this interval.
	
	Finally note that, in this example, the number of activations since the system start for nodes~1 and~3 is always the same at any timestamp, up to a difference of~1 (due to the jitters) because they have the same periods.
	In contrast, the number of activations for node~2 will be smaller than that of nodes~1 and~3 by an increasing difference, since node~2 is slower (period: 51 instead of~49).
	This phenomenon does not occur when periods are equal for all nodes, and makes this setting more challenging.
\end{example}
%------------------------------------------------------------

%------------------------------------------------------------
\begin{figure*}[!ht]
	\centering

		\begin{tikzpicture}[scale=.7, auto, ->, >=stealth']
		\footnotesize
		
		\tikzstyle{point}=[circle,fill,inner sep=1pt] % label=above:$a$
		\tikzstyle{barreemphase}=[-,densely dotted,black!50]
		
		% BARRE
		\draw[-] (0, 0) -- (0,3.5);
  
		% TIME
		\draw[thick] (0, 0) --++ (16, 0) node[right] {$t$};
		% VALEURS
		\foreach \x in {0, 10, ..., 150} % X
			\draw[-] (\x/10, 0) --++ (0, -.1) node [below] {\scriptsize{$\x$}};

		% NODE 1
		\draw[red] (0, 1) node[left] {$node_1$} --++ (16, 0);

		% NODE 2
		\draw[blue] (0, 2) node[left] {$node_2$} --++ (16, 0);
		
		% NODE 3
		\draw[green] (0, 3) node[left] {$node_3$} --++ (16, 0);
		
		% POINTS 1
		\node[point,red,label=above right:$t_1^0$] at (0, 1) (t10){};
		\node[point,red,label=above right:$t_1^1$] at (4.95, 1) (t11){};
		\node[point,red,label=above right:$t_1^2$] at (9.75, 1) (t12){};
		\node[point,red,label=above right:$t_1^3$] at (14.75, 1) (t13){};
% 		\draw[barreemphase] (t10) --++ (0, -1);
		\draw[barreemphase] (t11) --++ (0, -1);
		\draw[barreemphase] (t12) --++ (0, -1);
		\draw[barreemphase] (t13) --++ (0, -1);
		
		% POINTS 2
		\node[point,blue,label=above right:$t_2^0$] at (3, 2) (t20){};
		\node[point,blue,label=above right:$t_2^1$] at (8.1, 2) (t21){};
		\node[point,blue,label=above right:$t_2^2$] at (13.21, 2) (t22){};
		\draw[barreemphase] (t20) --++ (0, -2);
		\draw[barreemphase] (t21) --++ (0, -2);
		\draw[barreemphase] (t22) --++ (0, -2);
		
		% POINTS 3
		\node[point,green,label=above right:$t_3^0$] at (0.01, 3) (t30){};
		\node[point,green,label=above right:$t_3^1$] at (4.86, 3) (t31){};
		\node[point,green,label=above right:$t_3^2$] at (9.84, 3) (t32){};
		\node[point,green,label=above right:$t_3^3$] at (14.76, 3) (t33){};
		\draw[barreemphase] (t30) --++ (0, -3);
		\draw[barreemphase] (t31) --++ (0, -3);
		\draw[barreemphase] (t32) --++ (0, -3);
		\draw[barreemphase] (t33) --++ (0, -3);
		
	\end{tikzpicture}
	
	\caption{Activation of three nodes with uncertain periods and jitters}
	\label{figure:chronogramme}
\end{figure*}
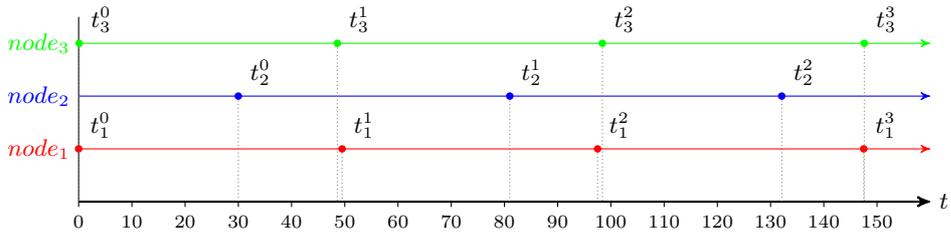
%------------------------------------------------------------

%------------------------------------------------------------
\begin{remark}\label{remark:constants}
	The rest of this paper assumes the constant values given in \cref{table:constants}.
	However, our method remains generic for constants of the same order of magnitude.
	Here, the variability of the periods is reasonably limited (around 4\,\%).
	A variability of more than 40\,\% will endanger the soundness of our method, as our upcoming assumption that between any three consecutive activations of a node, all nodes execute at least once, would not hold anymore.
\end{remark}
%------------------------------------------------------------

%%%%%%%%%%%%%%%%%%%%%%%%%%%%%%%%%%%%%%%%%%%%%%%%%%%%%%%%%%%%
\subsection{IDs, modes, messages}
%%%%%%%%%%%%%%%%%%%%%%%%%%%%%%%%%%%%%%%%%%%%%%%%%%%%%%%%%%%%
We assume that all the IDs of the nodes in the network are different.
Each node executes the same code.
Each node has the ability to send messages to all the nodes in the network, and can store (at least) one message received from any other node in the network. 
%Each node can also store information to be read by itself at another activation.
Nodes are either in mode $\modeOn{}$ and execute the code at each activation time, or do nothing when they are in mode~\modeOff{}.
(This models the fact that some nodes in the network might fail.)
A node in mode $\modeOn{}$ is in one of the following states:
\begin{itemize}
	\item \statusFollower{}: the node is not competing to become leader;
	\item \statusCandidate{}: the node is competing to become leader;
	\item \statusLeader{}: the node has declared itself to be the leader.
\end{itemize}
Each transmitted message is of the form: $message = (SenderID, state)$ where $state$
is the state of the sending node.

%%%%%%%%%%%%%%%%%%%%%%%%%%%%%%%%%%%%%%%%%%%%%%%%%%%%%%%%%%%%
\subsection{The algorithm}
%%%%%%%%%%%%%%%%%%%%%%%%%%%%%%%%%%%%%%%%%%%%%%%%%%%%%%%%%%%%

%------------------------------------------------------------
\begin{algorithm}[tb]

	\If{$node_i.EvenActivation $}{
		\nllabel{algo:node:if:begin}
		 $allMessages \leftarrow ReadMailbox()$
		 
		$\higherIDreceived{} \leftarrow \BFalse$ \nllabel{algo:node:higher:begin}
		
		\ForEach{$message \in allMessages$}{
			\If{$message.\SenderID{} > node_i.id$}{
				$state_{next} \leftarrow \statusFollower{}$ \nllabel{algo:node:higher:follower}
				
				$\higherIDreceived{} \leftarrow \BTrue{}$
			} % END IF
			\nllabel{algo:node:higher:end}
		} % END FOR
		\If{$\neg\ \higherIDreceived{}$}{
			\uIf{$node_i.state = \statusFollower{}$}{
				$state_{next} \leftarrow \statusCandidate{}$ \nllabel{algo:node:upgrade:follower}
			} % END IF
			
			\uElseIf{$node_i.state = \statusCandidate{}$}{
				$state_{next} \leftarrow \statusLeader{}$ \nllabel{algo:node:upgrade:candidate}
			} % END IF
			
			\ElseIf{$node_i.state = \statusLeader{}$}{
				$state_{next} \leftarrow \statusLeader{}$ \nllabel{algo:node:upgrade:leader}
			} % END IF
			
		} % END IF
		
		$node_i.state \leftarrow state_{next}$
		
		\nllabel{algo:node:if:end}
	} % END IF
	$node_i.EvenActivation \leftarrow \neg node_i.EvenActivation$ \nllabel{algo:node:swap}
	
	$message = \{node_i.id; node_i.state\}$ \nllabel{algo:node:prepare-message}
	
	$Send\_To\_All\_Network(message)$ \nllabel{algo:node:broadcast}

	\caption{$UpdateNode(i)$}
	\label{algo:UpdateNode}
\end{algorithm}
%------------------------------------------------------------

At each new activation, $node_i$ executes the code given in \cref{algo:UpdateNode}.
% This code only executes every two rounds, as one can see from the $EvenActivation$ Boolean flag.
%
In short, if the Boolean flag $node_i.EvenActivation$ (which we can suppose being initially arbitrary) is true, then the code \cref{algo:node:if:begin}--\cref{algo:node:if:end} is executed.
In this code, the node first reads its mailbox, and checks whether any message contains a higher node ID than the node ID (\cref{algo:node:higher:begin}--\cref{algo:node:higher:end}) and, if so, sets itself as a follower (\cref{algo:node:higher:follower}).
If no higher ID was received, the node ``upgrades'' its status from follower to candidate (\cref{algo:node:upgrade:follower}), from candidate to leader (\cref{algo:node:upgrade:candidate}), or remains leader if already leader (\cref{algo:node:upgrade:leader}).

Finally (and this code is executed at every iteration), the node swaps the Boolean flag $EvenActivation$ (\cref{algo:node:swap}), prepares a message with its ID and current state (\cref{algo:node:prepare-message}) and sends this message to the entire network (\cref{algo:node:broadcast}).
We assume that the $Send\_To\_All\_Network$ function sends a message to all nodes---including the sender.

% \begin{figure} %% add algorithm and algorithmic packages 
% \begin{algorithmic}%[1]
% \IF{$node_i.EvenActivation = \BTrue{}$}
% 	\STATE $allMessages \leftarrow ReadMailbox()$
% 	\STATE $\higherIDreceived{} \leftarrow \BFalse$
% 	\FOR{$message \in allMessages$} 
% 		\IF{$message.\SenderID{} > node_i.id$}
% 			\STATE $state_{next} \leftarrow \statusFollower{}$
% 			\STATE $\higherIDreceived{} \leftarrow \BTrue{}$
% 		\ENDIF
% 	\ENDFOR
% 	\IF{$\neg\ \higherIDreceived{}$}
% 		\IF{$node_i.state = \statusFollower{}$}
% 			\STATE $state_{next} \leftarrow \statusCandidate{}$
% 		\ENDIF
% 		\IF{$node_i.state = \statusCandidate{}$}
% 			\STATE $state_{next} \leftarrow \statusLeader{}$
% 		\ENDIF
% 		\IF{$node_i.state = \statusLeader{}$}
% 			\STATE $state_{next} \leftarrow \statusLeader{}$
% 		\ENDIF
% 	\ENDIF
% 	\STATE $node_i.state \leftarrow state_{next}$
% \ENDIF
% \STATE $node_i.EvenActivation \leftarrow \neg node_i.EvenActivation$
% \STATE $message = \{node_i.id; node_i.state\}$
% \STATE $Send\_To\_All\_Network(message)$
% \end{algorithmic}
% \caption{$UpdateNode(i)$}
% \label{algo:UpdateNode}
% \end{figure}
%
We can see that the significant part of the code (\cref{algo:node:if:begin}--\cref{algo:node:if:end}) is only executed once every two activations (due to Boolean test $node_i.EvenActivation $).
This is enforced in order to ensure that each node executes the code after receiving at least one message from all the other nodes (in mode $\modeOn{}$). 
%Such a mechanism intends to ensure a relaxed form
%of  synchrony, in the sense that, at each time $t$, the number of performed activations is always {\em almost the same}
%(\ie{} is equal $\pm$ a certain bounded number), 
%for all the nodes.
However, note that each node sends a message at each iteration.

The order of magnitude of the constants in \cref{table:constants} gives the immediate lemma.

%------------------------------------------------------------
\begin{lemma}\label{lemma:2}
	Assume a node~$i$ and activation times $t_i^{j}$ and $t_i^{j+2}$.
	Then in between these two activations, node~$i$ received at least one message from all nodes.
\end{lemma}
%------------------------------------------------------------
\begin{proof}
	From \cref{table:constants,algo:UpdateNode}.
\end{proof}
%------------------------------------------------------------

%------------------------------------------------------------
\begin{remark}
	For different orders of magnitudes, we may need to execute the code once every more than two activations.
	For example, if we set $\jitterMin{} = - 25$ and $\jitterMax{} = 25$ in \cref{table:constants}, the code should be executed every three activations for our algorithm to remain correct.
% 	Let us give some insights on the reasons why (most of) the code is executed once every two activations.
% 	Recall that there are \emph{jitters} and, most importantly, that the periods can be slightly different \ea{je suis là}
% 	\rs{On peut aussi mettre de manière plus claire que le fait de lire les messages et de calculer qu’une fois sur 2 est justement ajouté pour palier à ceci et que si les bornes des périodes, jitter min et max changent, il faudrait peut être lire et calculer qu’une fois sur 3, sur 4, ou sur 10 etc.}
\end{remark}
%------------------------------------------------------------

%%%%%%%%%%%%%%%%%%%%%%%%%%%%%%%%%%%%%%%%%%%%%%%%%%%%%%%%%%%%
\subsection{Objective}
%%%%%%%%%%%%%%%%%%%%%%%%%%%%%%%%%%%%%%%%%%%%%%%%%%%%%%%%%%%%

We first introduce the following definitions.

%------------------------------------------------------------
\begin{definition}[round]
	A {\em round} is a time period during which all the nodes that are $\modeOn{}$ have sent at least one message.
\end{definition}
%------------------------------------------------------------

%------------------------------------------------------------
\begin{definition}[cleanness]\label{definition:cleanness}
	A round is said to be {\em clean} if during its time period no node have been switched from $\modeOn{}$ to \modeOff{} or from \modeOff{} to $\modeOn{}$.
\end{definition}
%------------------------------------------------------------

The correctness property that we want to prove is:

	\smallskip

\noindent\fcolorbox{black}{blue!15}{
	\begin{minipage}{.95\columnwidth}
		``{\bfseries When, following a preliminary clean round, 4 new clean rounds occur, 
the node with the highest ID is recognized as the leader by all the nodes in modes $\modeOn{}$ of the network.}''
	\end{minipage}
}
	
	\smallskip

This property is denoted by (P) in the following.

%------------------------------------------------------------
\begin{remark}[fault model]
	Our model does allow for faults but, according to \cref{definition:cleanness}, only prior to the execution of the algorithm.
	That is, once it has started, all nodes remain in \modeOn{} or \modeOff{} during its entire execution.
	If in reality there is a fault during the execution, it suffices to consider the execution of the algorithm at the next clean round.
\end{remark}
%------------------------------------------------------------

%%%%%%%%%%%%%%%%%%%%%%%%%%%%%%%%%%%%%%%%%%%%%%%%%%%%%%%%%%%%
%%%%%%%%%%%%%%%%%%%%%%%%%%%%%%%%%%%%%%%%%%%%%%%%%%%%%%%%%%%%
\section{Direct verification of the leader election algorithm}\label{section:direct}
%%%%%%%%%%%%%%%%%%%%%%%%%%%%%%%%%%%%%%%%%%%%%%%%%%%%%%%%%%%%
%%%%%%%%%%%%%%%%%%%%%%%%%%%%%%%%%%%%%%%%%%%%%%%%%%%%%%%%%%%%
%In this part how we verified the algorithm on a large number of nodes.

In this section, we first verify our algorithm for a fixed number of processes.

We describe here the results obtained by %an SMT solver
	SafeProver
on a model~$M$ representing directly a network of a fixed, constant number of~$p$ processes (without abstraction); for a small number~$p$ of nodes,
we thus obtain a simple proof of the correctness of the algorithm. 
%\subsection{Direct modeling (without abstraction)}
The model includes explicitly a representation of each node of $\mathcal{N}$ 
as well as their associated periods, first activation times, local memories, and mailboxes of received messages.
The code is given in \cref{algo:modelM}.
The mailbox is represented as a queue, initially filled with a message from oneself.\footnote{%
	An initial empty mailbox would do as well but, in the actual Thales system, this is the way the initialization is performed.
}

%------------------------------------------------------------
\begin{algorithm}[tb]

	$Activation[1,\ldots,p] \leftarrow [0,\ldots,0]$
	
	\tcp{Network initialization}
	%% \STATE $networkStatus \in \{0,1\}^{p}$ \% This free variable models the nodes that are $\modeOn{}$ and {\em Off}

	\ForEach{$i\in\{1,\ldots,p\}$}{
	%%	\IF{$networkStatus(i)$}
			$node_{i}.id \in \grandn$
			
			$node_{i}.per \in [\periodMin{}; \periodMax{}]$
			
			$node_{i}.start \in [0; node_{i}.per]$ \ea{Romain avait mis 1, j'ai remplacé par 0}
			
			$node_{i}.state \in \{\statusFollower{},\statusCandidate{},\statusLeader{}\}$
			
			$node_{i}.EvenActivation \in \{\BTrue{},\BFalse\}$
			
					$node_{i}.mode \in \{\modeOn{}, \modeOff{}\}$
					
		$nextActivationTime(i) \leftarrow node_{i}.start$
	%% \ENDIF
	} % END FOREACH

	\tcp{Mailboxes initializations}
	\ForEach{$i\in\{1,\ldots,p\}$}{
% 		$message_{i} \in  \{node_{i}.id\} \times  \{\statusFollower{},\statusCandidate{},\statusLeader{}\}$
		
		\tcp{Arbitrary mailbox initialization with a message from oneself}
		$node_{i}.mailbox \leftarrow [(node_{i}.id , \statusFollower{})] $
	} % END FOREACH
	\ForEach{$i\in\{1,\ldots,p\}$}{
		\ForEach{$j\in\{1,\ldots,p\}$}{
			\If{$node_j.mode = \modeOn{}$}{
% 					$node_{i}.mailbox \leftarrow \{node_{j}.mailbox, message_{j}\}$
					$node_{i}.mailbox.enqueue(message_{j})$
			} % END IF
		} % END FOREACH
	} % END FOREACH

	\tcp{Main algorithm}
	\While{$\BTrue$}{
		$i \leftarrow \argmin(nextActivationTime)$
		
		\If{$node_{i}.mode = \modeOn{}$}{
			$UpdateNode(i)$
			
				$Activation(i) \leftarrow Activation(i)+1$
				
			$jitter \in [\jitterMin{}, \jitterMax{}]$
			
			$nextActivationTime(i) \leftarrow nextActivationTime(i) + node_{i}.per + jitter$
		} % END IF
	} % END WHILE

	\caption{SafeProver code for model~$M$}
	\label{algo:modelM}
\end{algorithm}

During the initialization declaration, we set everything as free variables (with some constraints, \eg{} on the periods) in order to have
no assumptions on the state of the network at the beginning.
This ensures that this model is valid whatever happened in the past, and this can be seen as a symbolic initial state: this notion of symbolic initial state was used to solve a challenge by Thales~\cite{SAL15}, also featuring uncertain periods.
We also fully initialize the mailboxes of all the nodes since we are assuming that we are right after a clean round.
The variable $Activation$ is used as a variable to store how many times a node has been executed after the last clean round.
The code of called function $UpdateNode(i)$ corresponds exactly to \cref{algo:UpdateNode}.
%With some rework on what the nodes exchange in the
%messages and what the nodes store in memory, we can apply this technique to try to verify another algorithm.

The property (P) we want to prove is formalized as:\\
%\begin{itemize}
$ \big(\forall i \in \{1,\ldots,p\}, Activation(i) \geq 4\big)\ \Rightarrow$\\
$\hspace*{4mm}\big(\forall j\in\{1,\ldots,p\}, j \neq maxId: node_j.state = \statusFollower{} $\\
\hspace*{6mm}$\ \land\ node_{maxId}.state = \statusLeader{}\big)$\\
with $ maxId = \argmax(\{node_i.id \mid node_i.mode = \modeOn{} \}_{i\in\{1,\ldots,p\}})$.
%\end{itemize}
Using this model and SafeProver \cite{Etienne:2017}, we obtain 
the proof of (P) with the times tabulated in~\cref{table:SMT}.\footnote{%
	All the experiments reported in this paper have been run on a machine with two Intel$^{\textregistered}$ Xeon$^{\textregistered}$ CPU E5-2430 at~2.5\,GHz,
with 164\,GiB of RAM  and running a Debian~9 Linux distribution.
\todo{page Web avec URL !!! \url{https://www.imitator.fr/static/VMCAI19/}}
}
While this method allows us to formally prove the leader election for up to~5 nodes, SafeProver times out for larger number of nodes.
This leads us to consider another method to prove the correctness of our algorithm for larger numbers.
 
%In the next part, we describe how one can use an abstraction-based method
%in order to gain
%several order of magnitudes in the number of nodes for which we can prove
%that the leader election algorithm is correct.

\begin{table}[tb]
	\centering
	\caption{Computation times}
	\begin{tabular}{| c | r |}
%\item for $p=3$:\ \ $xx.xx$s,
	\hline
		\cellHeader{Nodes} & \cellHeader{Time (s)}\\
	\hline
	$p=4$ & $66.65$\\ 
	\hline
	$p=5$ & $215.61$\\
	\hline
	$p=6$ & time out ($> 3600$)\\
	\hline
	\end{tabular}

	\label{table:SMT}
\end{table}

\section{Abstraction-based method}\label{section:abstraction}
We now explain how to construct an abstract model $\Mabs{}$ of the original model~$M$.
This model $\Mabs{}$ clusters together all the~$p$ processes, except the process $node_i$ under study (where $i$ is arbitrary, \ie{} a free variable); $\Mabs{}$ also abstracts away the timing information contained in~$M$.
We then use SaveProver to infer two basic properties P1 and P2 for $\Mabs{}$.

In a second phase, we consider an auxiliary simple (abstract) model $T$ 
of~$M$ which merely contains relevant timing information; we then use a parametric timed model checker to infer a third property (P3) for~$T$.
The parametric timed model checker is required due to the \emph{uncertain} periods, that can have any value in $[\periodMin{} , \periodMax{}]$ but remain constant over the entire execution.

In the third phase, we consider again the model $\Mabs{}$, and integrate P1--P3 to SafeProver as assumptions, which allows us to infer a fourth property~P4.
The properties P1 and P4 express together a statement equivalent to the 
desired correctness property P of the leader election algorithm. 
The advantage of reasoning with abstract models $\Mabs{}$ and $T$ 
rather than directly to~$M$, is to prove~P for a large number~$p$ of processes.

% We now give more details on the construction of $\Mabs{}$ and $T$ and the proof of P1--P4.
We now describe our method step by step in the following.

\subsection{Abstract model $\Mabs{}$ and proof of P1-P2}
The idea is to model the system as one node $node_i$
(the node of interest) interacting with the rest of the network:
%The model $\Mabs{}$ 
%is made of two parts. One is composed by the node under study 
$node_i$ receives messages from the other nodes which are clustered into a single abstract process (see \cref{figure:abstraction}).
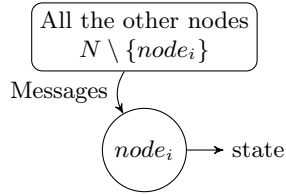
\begin{figure}[tb]
	\centering 
	\begin{tikzpicture}[scale=1.5, auto, ->, >=stealth']
		\node[draw, rounded corners, align=center] at (0, 1) (all) {All the other nodes\\$N \setminus \{ node_i \}$};
 
		\node[draw, circle] at (0, 0) (nodei) {$node_i$};
		
		\node[] at (1, 0) (state) {state};
		\path (all) edge[bend right] node[left,align=center]{Messages} (nodei);
		\path (nodei) edge[] (state);
	\end{tikzpicture}
	%\caption{Node}
	\caption{Scheme of model $\Mabs{}$ with node $i$ under study interacting with the cluster of all the other nodes}
	
	\label{figure:abstraction}
\end{figure}
In the abstract model $\Mabs{}$, each node can take any state at any activation, with no regards to the parity ($node_i.EvenActivation$), what has been previously sent, what $node_i$ is sending.
%This represents an obvious over-approximation of the real behavior of the nodes.
We only consider the activation of $node_i$.
The rest of the network is abstracted by the messages contained in the mailbox of~$node_i$. 
Since we assume that at least one clean round has passed, we always have a message from a working node in the mailbox.
The code is given in \cref{algo:abstractmodelM}.
(Note its analogy with the SafeProver code of \cref{algo:modelM} for~$M$.)
%In \ref{algo:modular}, we present the SMT model. 
The first four lines define free variables.

The list of assumptions that the solver can make on the messages received is denoted by $List\_of\_assumptions\_on\_message_j$.
This list is initially empty, and augmented with ``guarantees'' (\aka{} ``proven properties'') on the other nodes as they are iteratively generated by the solver.
%
%
% BEGIN ÉA j'ai remplacé i par j
The first proven  properties are:
\begin{itemize}
	\item P1: $ (Activation(j) \geq 2 \ \land\  node_j.id \neq maxId)  \Rightarrow\ node_j.state = \statusFollower{}$
	\item P2: $ (Activation(j) \geq 2 \ \land\  node_j.id = maxId) $\\ 
\hspace*{\fill} $\Rightarrow node_j.state \in \{\statusCandidate{}, \statusLeader{}\}$
\end{itemize}\ea{note, j'ai remplacé des variables un peu au hasard}%
% END ÉA j'ai remplacé i par j
In these properties, $j$ is a free variable: therefore, it is ``fixed'' among one execution, but can correspond to any of the node IDs.
The two properties state that, after two rounds, a node which has not the largest ID is necessarily a follower (P1), or a candidate or a leader if it has the largest ID (P2).
As said before, P1 and P2 are then integrated to $List\_of\_assumptions\_on\_message_j$.

$Guarantee\_to\_prove$ contains iteratively P1, then P1 and~P2, then P1, P2 and~P4.

%------------------------------------------------------------
\begin{algorithm}[tb]
	%$Activation[1,\ldots,p] \leftarrow [0,\ldots,0]$
	%$Activation{1,\ldots,p}$
	Assume $i \in \{1,\ldots,p\}$
	
	%$Activation(i) \leftarrow 0$
	Assume $node_{i}.id \in \grandn$
	
	Assume $node_{i}.state \in \{\statusFollower{},\statusCandidate{},\statusLeader{}\}$
	
	Assume $node_{i}.EvenActivation \in \{\BTrue{}, \BFalse\}$
	
	$Activation(i) \leftarrow 0$
	
	%\ENDFOR ???
	%\% Mailboxes initializations
	\While{$\BTrue$}{
		\For{$j\in\{1,\ldots,p\}$}{
			$message_{j} \in  \{node_{j}.id\} \times  \{\statusFollower{},\statusCandidate{},\statusLeader{}\}$
	%		\textbf{Assume:} \% property P3\\
	%\hspace*{\fill} $Activation(j) \in [Activation(i) - 1; Activation(i) + 2]$\\
			
			\textbf{Assume:}\  $List\_of\_assumptions\_on\_message_{j}$
			
% 			$node_{i}.mailbox \leftarrow \{node_{j}.mailbox, message_{j}\}$
			$node_{i}.mailbox.enqueue(message_{j})$
		} % END FOR

		$UpdateNode(i)$
		
		$Activation(i) ++ $
		
		$Guarantee\_to\_prove$
	} % END WHILE

	\caption{SafeProver code for abstract model $\Mabs{}$}
	\label{algo:abstractmodelM}
\end{algorithm}

\subsection{Abstract model $T$ and proof of P3}\label{ss:abstraction2}
%Using \imitator{}, we model the activation patterns of two nodes.

To represent the timed abstract model $T$ of~$M$, we use an extension of the formalism of timed automata~\cite{AD94}, a powerful extension of finite-state automata with clocks, \ie{} real-valued variables that evolve at the same time.
These clocks can be compared with constants when taking a transition (``guards''), or to remain in a discrete state (``invariants'').
Discrete states are called \emph{locations}.
Timed automata were proven successful in verifying many systems with interactions between time and concurrency, especially with the state-of-the-art model-checker \uppaal{}~\cite{Kim}.
However, timed automata cannot model and verify arbitrary periods: while it is possible to model a different period at each round, it is not possible to first fix a period once for all (in an interval), and then use this period for the rest of the execution.
We therefore use the extension ``parametric timed automata''~\cite{AHV93,Andre17STTT} allowing to consider \emph{parameters}, \ie{} unknown constants (possibly in an interval).
\imitator{}~\cite{AFKS12} is a state-of-the-art model checker supporting this formalism.

In our method, the timed abstract model $T$ of~$M$ is a product of two similar parametric timed automata representing the node $i$ under study and 
a generic node $j$ belonging to ${\cal N}\setminus\{i\}$ respectively.
%Each of the two nodes is represented by a timed automaton. 
%In \cref{fig:5}, we present the template of this automaton.
Each parametric timed automaton contains a single location.
The parametric timed automaton corresponding to $node_i$ uses an activation period $per_i$ that we model as a parameter.
Indeed, recall that the period belongs to an interval: taking a value in the interval at each round would not be correct, as the period would not be constant.
This is where we need parameters in our method.
In addition, we constrain this parameter $per_i$ to belong to $[\periodMin{},\periodMax{}]$.
Each automaton has its own clock $c_i$ that is used to measure how much time has passed since the last activation.
Each automaton has a discrete variable\footnote{%
	Discrete variables are global Boolean- or integer-valued variables, that can be read or written in transition guards.
	If their domain is finite they are syntactic sugar for a larger number of locations.
} $Activation(i)$ which is initialized at 0 and is used to count the number of activations for this node.
We give the constraint on $c_i$ at the beginning that $c_i \in [0, per_i + \jitterMax{}]$.
An activation can occur as soon as $c_i$ reaches
$per_i + \jitterMin{}$. This is modeled by the guard $c_i \geq per_i + \jitterMin{}$ on the transition that resets $c_i$
 and increment $Activation(i)$.
 An activation can occur as long as $c_i$ is below or equal to $per_i + \jitterMax{}$. This is modeled by the invariant 
 $c_i \leq per_i + \jitterMin{}$ on the unique location of the automaton. This invariant forces the transition to occur when $c_i$ reaches its upper bound.
This parametric timed automaton is represented in \cref{fig:5}.\footnote{%
	The color code is that of \imitator{} automated \LaTeX{} outputs: clocks are in blue, parameters in orange, and discrete variables in pink.
}
The other component representing the cluster of the rest of the nodes is modeled similarly as a generic component $node_j$.

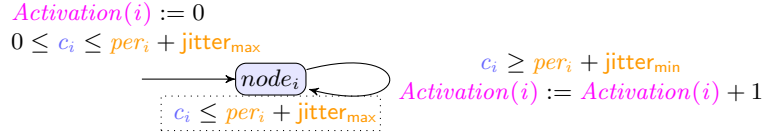
\begin{figure*}[tb]
	\centering

		\begin{tikzpicture}[scale=3, auto, ->, >=stealth']
 
		\node[location, initial] at (0,0) (l1) {$node_i$};
		% Invariant of location l1
		\node [invariant,below] at (l1.south) {\begin{tabular}{@{} c @{\ } c@{} }& $ \styleclock{c_i} \leq \styleparam{per_{i}} + \styleparam{\jitterMax{}}$\\\end{tabular}};

		\node[above left,align=left] at (l1.north) {$\styledisc{Activation(i)} := 0$\\$0 \leq \styleclock{c_i} \leq \styleparam{per_{i}} + \styleparam{\jitterMax{}}$};
 
		\path (l1) edge[loop right] node[right,align=center]{$\styleclock{c_i} \geq \styleparam{per_{i}} + \styleparam{\jitterMin{}}$ \\ $\styledisc{Activation(i)} := \styledisc{Activation(i)} + 1$} (l1);
	\end{tikzpicture}
	
	\caption{Component~1 of timed model~$T$} %  ($exec_i$ represents here the number of activations of~$node_i$) % NOTE(ÉA: 11/10/2018 du coup je remplace)
	\label{fig:5}
\end{figure*}

For nodes $node_i$ 
and $node_j$, the property  that we want to specify
corresponds in the direct model~$M$
(without abstraction) of \cref{section:direct} to:
\begin{itemize}
\item %Q:
$(Activation(i) \leq 13 \ \land\  Activation(j) \leq 13)$\\ 
\hspace*{\fill} $\Rightarrow\ \mid Activation(i) - Activation(j) \mid \ \leq\  2$.
\end{itemize}

In our timed abstract model $T$, such a property becomes:
\begin{itemize} \item
(P3):
$\forall i\in\{1,\dots,p\}\ Activation(j) \leq 13 \Rightarrow $\\
\hspace*{1mm} 
$ -2 \leq Activation(j) - Activation(i) \leq 1$.
\end{itemize}
where $Activation(i)$ denotes the number of activations of node~$i$ since the last clean round.

The value ``13'' has been obtained experimentally: smaller numbers led the algorithm to fail (the property was not satisfied).
Intuitively, it consists in the number of activations by which we are sure the leader will eventually be elected.

The proof of  P3 is obtained by adding to the model an observer\footnote{%
	An observer is an additional automaton that can synchronize with the system (using synchronized actions, clocks or discrete variables values), without modifying its behavior nor blocking it.
	See \eg{} \cite{ABBL98,Andre13ICECCS}.
} automaton checking the value of the discrete variables $Activation(i)$ and $Activation(j)$, which goes to a {\em bad} location when the property is violated.
The property is then verified
by showing  that the {\em bad} location is not reachable.
For the values of the timing constants\ea{Romain avait mis: ``$\periodMin{} = 39000$, $\periodMax{} = 41000$, $\jitterMin{} = -500$, and $\jitterMax{} = 500$,'' Mais non, pourquoi avoir fait un rescaling???}
in \cref{table:constants},
\imitator{} proves P3 (by showing the non-reachability of the bad location) in $12$\,s.
Recall that, thanks to our assumption on the number of nodes, we only used two nodes in the input model for \imitator{}.

In the next part, we show how the addition of~P3 as an assumption in the 
original abstract model allows to prove the desired property~P for a large number of nodes.

\subsection{Proof of P using P1--P3 as assumptions}

In addition to P1-P2,
we now put P3 ($Activation(j) \in [Activation(i) - 1; Activation(i) + 2]$) as an element of $List\_of\_assumptions\_on\_message_j$ used
in the SafeProver code of $\Mabs{}$ (see \cref{algo:abstractmodelM}).
SafeProver is then able to generate the following property:
\[ P4 : (Activation(i) \geq 4 \ \land\  node_i.id = maxId) \Rightarrow node_i.state = \statusLeader{} \]
Property~P4 states that the node
with the highest ID will declare itself as \statusLeader{} after at most 4 activations.
Besides, property~P1 states that a node, not having the highest ID, is in the state $\statusFollower{}$ within at most 2 activations.  
Properties~P1 and~P4 together thus express a statement equivalent to the desired correctness property~P.
The global generation of properties (P1), (P2) and (P4) by SafeProver takes the computation times tabulated in \cref{table:SMT2}.
As one sees, the computation time is now smaller by an order of magnitude than the ones given in \cref{table:SMT}, thus showing the good scalability of our method.

\begin{table}[tb]
	\centering
	\caption{Computation times}
	\begin{tabular}{| r | r |}
%\item for $p=3$:\ \ $xx.xx$s,
	\hline
		\cellHeader{Nodes} & \cellHeader{Time (s)}\\
	\hline
	$p = 500$ & $13.34$\\ 
	\hline
	$p = 1000$ & $45.95$\\
	\hline
	$p = 5000$ & $623.46$\\
	\hline
	\end{tabular}

	\label{table:SMT2}
\end{table}

\begin{remark}
% BEGIN SUPPRIMÉ
% 	Note that verifying the model for~5,000 nodes does not a priori guarantee that the verification applies to less than 5000 nodes.
% 	This is not a real issue for two reasons:
% 		first, the exact number of nodes on which the algorithm is executed in the industrial setting is known before hand, so it suffices to verify it for that number.
% 		Second, verifying for all numbers of nodes below some threshold would be of course longer, but this could be easily parallelized (\eg{} on a cluster with up to~5000 nodes).
% 	Finally, note that, in our original industrial algorithm, our model also featured the possibility of nodes to fail (which we omitted here for sake of simplification); therefore, in that model enriched with failures, verifying the correctness for 5000~nodes is equivalent to verifying it for any smaller number, as any number of nodes can fail (and never recover) which is equivalent to smaller models.
% END SUPPRIMÉ
	Note that verifying the model for~5,000 nodes also gives a guarantee for any smaller of nodes.
	Indeed, we can assume that an arbitrary number of nodes are in mode~$\modeOff{}$, and remain so, which is equivalent to a smaller number of nodes.
\end{remark}

\subsection{Discussion}

\paragraph{Soundness}
We briefly discuss the soundness of the algorithm.
First, note that the assumptions used above have been validated by the system designers (\ie{} those who designed the algorithm).
Second, SafeProver validated the assumptions, \ie{} proved that they were not inconsistent with each other (which would result in an empty model).

Now, the abstraction used in \cref{ss:abstraction2}, \ie{} to consider only two nodes, is the one which required most human creativity.
Let us briefly discuss its soundness.
Our abstraction allows to model the sending of any message, which includes the actual message to be sent in the actual system.
The fact that a message was necessarily received in the actual system between two (real) executions of the node under study is given by the fact that all nodes necessarily execute at least once in the last two periods (see \cref{lemma:2}).
% ROMAIN (10/2018)
% Alors, les assumptions sont validées par les spécifieurs du système et donc de l’algorithme. Ensuite l’ensemble des assumptions sont vérifiées comme n’étant pas incohérentes (et donc d’avoir un modèle vide) dans Sprover. Ensuite, l’abstraction du modele en 4.2 est sound pour moi car l’abstraction des autres nœuds permet l’envoi de n’importe quel message (ce qui inclus le message qui aurait du être envoyé). Le fait qu’un message soit forcément reçu entre deux (vraies) exécutions du node under study est donnée par le fait que tout le monde s’éxécute dans les deux dernières périodes. C’est pas mal de justifications « à la main » malheureusement, on avait prévu de bosser plus sérieusement dessus mais dans le dossier de qualification, ca a suffit apparemment donc on ne nous a jamais demandé de travailler dessus. Peut être un future work à mentioner ?
Of course, this soundness is only valid under our own assumptions on the variability of the period, considering the constants in \cref{table:constants}: if the period of one node is~1 while the other is~100, our framework is obviously not sound anymore.

\todo{clear overabstraction due to the abstract node allowing everything}

\paragraph{Parametric vs.\ parametrized model checking}
As shown in \cref{table:SMT2}, we verified the model for a \emph{constant} number of nodes.
This comes in contrast with the recent work on parameterized verification (\eg{} \cite{ADRST16,CDZ17,CDZ17b}).
However, while these latter consider a parameterized number of nodes, they consider \emph{non-parametric} timed models; in contrast, we need here parametric timed models to be able to represent the uncertainty on the periods.
Combining both worlds (parameterized number of nodes with parametric timed models) would be of interest---but remains a very challenging objective.

%%%%%%%%%%%%%%%%%%%%%%%%%%%%%%%%%%%%%%%%%%%%%%%%%%%%%%%%%%%%
%%%%%%%%%%%%%%%%%%%%%%%%%%%%%%%%%%%%%%%%%%%%%%%%%%%%%%%%%%%%
\section{Conclusion}\label{section:conclusion}
%%%%%%%%%%%%%%%%%%%%%%%%%%%%%%%%%%%%%%%%%%%%%%%%%%%%%%%%%%%%
%%%%%%%%%%%%%%%%%%%%%%%%%%%%%%%%%%%%%%%%%%%%%%%%%%%%%%%%%%%%

We described a method combining abstraction, SafeProver and parametric timed model-checking in order to prove the correctness of a variant of an asynchronous leader election algorithm designed by Thales.
Our approach can efficiently verify the leader election after a fixed number of rounds for a large number~$p$ of processes (up to $p=5000$).
The method relies on the construction of two abstract models $\Mabs{}$ and $T$ of the original model~$M$.
Although it is intuitive, it could be interesting to prove formally that each abstraction is {\em correct} in the sense that it {\em over-approximates} all the behaviors of~$M$.

\subsection*{Perspectives}
Many variants of the algorithm can be envisioned (loss of messages, non-instantaneous transmission, non-complete graph topology, \dots).
The fault model could also be enriched.
It will then be also interesting to propose extensions of our abstraction-based method to prove the correctness of such extensions.

The correctness of the method relies on the order of magnitude of the constants used (\cref{remark:constants}).
% BEGIN Copié-collé d'un reviewer
For different intervals, it might be necessary to both adapt the algorithm (read messages only every $k$ activations) but also the assumptions used in the proof using abstraction, a manual and possibly error-prone process.
A more general verification method would be desirable.
% END Copié-collé d'un reviewer

In addition, the number of activations in our correctness property (``after 13 activations, the leader is elected'') was obtained using an incremental verification (values of up to~12 all gave concrete counterexamples).
As a future work, we would like to automatically infer this value too, \ie{} obtaining the minimal value of activations before a leader is guaranteed to be elected.

Finally, adding probabilities to model the fault of nodes will be of interest.

\section*{Acknowledgment}

% This work is partially supported by Institut Farman (ENS Paris-Saclay $\&$ CNRS).
We thank anonymous reviewers for very useful remarks and suggestions.

\bibliography{bully}
%%%%%%%%%%%%%%%%%%%%%%%%%%%%%%%%%%%%%%%%%%%%%%%%%%%%%%%%%%%%%
%%%%%%%%%%%%%%%%%%%%%%%%%%%%%%%%%%%%%%%%%%%%%%%%%%%%%%%%%%%%%

\end{document}